\documentclass[11pt]{article}
\usepackage{fullpage} 

\usepackage[all=normal,bibliography=tight,bibnotes=tight]{savetrees}

\usepackage{amssymb,comment}
\usepackage{amsmath}
\usepackage{xspace}
\usepackage{tikz}
\usepackage{pgflibrarysnakes}
\usetikzlibrary{snakes}
\usepackage{times}
\usepackage{amstext}
\usepackage{calc}
\usepackage{amsopn}
\usepackage[noend]{algorithmic}
\usepackage[boxed]{algorithm}
\usepackage{eucal}
\usepackage{latexsym}
\usepackage{tabularx}

\usepackage{amsthm}

\newtheorem{theorem}{Theorem}

\newtheorem{lemma}[theorem]{Lemma}

\theoremstyle{definition}
\newtheorem{definition}[theorem]{Definition}

{\bf}{\rm}

\floatname{algorithm}{Algorithm}

\newcommand{\defproblemu}[4]{
  \vspace{1mm}
\noindent\fbox{
  \begin{minipage}{\textwidth}
  \begin{tabular*}{\textwidth}{@{\extracolsep{\fill}}lr} #1 & {\bf{Parameter:}} #3 \\ \end{tabular*}
  {\bf{Input:}} #2  \\
  {\bf{Question:}} #4
  \end{minipage}
  }
  \vspace{1mm}
}

\newcommand{\defproblemugoal}[4]{
  \vspace{1mm}
\noindent\fbox{
  \begin{minipage}{\textwidth}
  \begin{tabular*}{\textwidth}{@{\extracolsep{\fill}}lr} #1 & {\bf{Parameter:}} #3 \\ \end{tabular*}
  {\bf{Input:}} #2  \\
  {\bf{Goal:}} #4
  \end{minipage}
  }
  \vspace{1mm}
}

\newcommand{\gfvs}{{\sc{Group Feedback Vertex Set}}\xspace}
\newcommand{\gfvsshort}{{\sc{GFVS}}\xspace}
\newcommand{\compgfvs}{{\sc{Compression Group Feedback Vertex Set}}\xspace}
\newcommand{\cgfvsshort}{{\sc{C-GFVS}}\xspace}
\newcommand{\fvsname}{{\sc{Feedback Vertex Set}}}
\newcommand{\sfvs}{{\sc{Subset Feedback Vertex Set}}\xspace}
\newcommand{\esfvs}{{\sc{Edge Subset Feedback Vertex Set}}\xspace}
\newcommand{\esfvsshort}{{\sc{ESFVS}}\xspace}
\newcommand{\oct}{{\sc{Odd Cycle Transversal}}\xspace}
\newcommand{\octshort}{{\sc{OCT}}\xspace}
\newcommand{\fvsshort}{{\sc{FVS}}}
\newcommand{\mwc}{{\sc{Multiway Cut}}\xspace}
\newcommand{\mwcshort}{{\sc{MWC}}\xspace}
\newcommand{\asatname}{{\sc{Almost 2-SAT}}\xspace}
\newcommand{\Ohstar}{\ensuremath{O^\ast}}

\begin{document}

  \date{}

\author{
  Marek Cygan 
  \thanks{
    Institute of Informatics, University of Warsaw, Poland,
      \texttt{cygan@mimuw.edu.pl}
  }
  \and
  Marcin Pilipczuk
  \thanks{
    Institute of Informatics, University of Warsaw, Poland,
      \texttt{malcin@mimuw.edu.pl}
  }
  \and
  Micha\l{} Pilipczuk 
  \thanks{
    Department of Informatics, University of Bergen, Norway, \texttt{michal.pilipczuk@ii.uib.no}
  }
}

  \title{On group feedback vertex set parameterized by the size of the cutset}

  \maketitle

\begin{abstract}
We study the parameterized complexity of
a robust generalization of the classical \fvsname{} problem,
namely the \gfvs{} problem; we are given a graph $G$ with edges labeled
with group elements, and the goal is to compute the smallest set of vertices that
hits all cycles of $G$ that evaluate to a non-null element of the group.
This problem generalizes not only \fvsname{}, but also \sfvs, \mwc{} and \oct{}.
Completing the results of Guillemot [Discr. Opt. 2011], we provide a fixed-parameter
algorithm for the parameterization by the size of the cutset only.
Our algorithm works even if the group is given as a polynomial-time oracle.
\end{abstract}

\section{Introduction}

The parameterized complexity is an approach for tackling NP-hard problems by designing algorithms that perform well, when the instance is in some sense simple; its difficulty is measured by an integer, called the {\emph{parameter}}, additionally appended to the input. Formally, we say that a problem is {\emph{fixed-parameter tractable}} (FPT), if it admits an algorithm that given input of length $n$ and parameter $k$, resolves the task in time $f(k)n^c$, where $f$ is some computable function and $c$ is a constant independent of the parameter.

The search for fixed-parameter algorithms led to the development of a number of new techniques and gave valuable insight into structures of many classes of NP-hard problems. Among them, there is a family of so-called {\emph{graph cut}} problems, where the goal is to delete as few as possible edges or vertices (depending on the variant) in order to make a graph satisfy a global separation requirement. This class is perhaps best represented by the classical \fvsname{} problem (\fvsshort{}) where, given an undirected graph $G$, we seek for a minimum set of vertices
that hits all cycles of $G$. Another examples are \mwc{} (\mwcshort{}: separate each pair from a given set of terminals in a graph
with a minimum cutset) or \oct{} (\octshort{}: make a graph bipartite by a minimum number
of vertex deletions).

The research on the aforementioned problems had a great impact on the development
of parameterized complexity. 
The long line of research concerning parameterized algorithms for \fvsshort{} 
contains 
\cite{fvs:4krand,fvs1,fvs:3.83k,fvs:5k,fvs7,fvs2,fvs3,guo:fvs,fvs6,fvs5},
leading to an algorithm working in $3^k n^{O(1)}$ time \cite{fvs:3k}.
The search for a polynomial kernel for \fvsshort{} lead to surprising applications
of deep combinatorial results such as the Gallai's theorem \cite{fvs:quadratic-kernel},
which has also been found useful in designing FPT algorithms \cite{sfvs}.
While investigating the graph cut problems such as \mwcshort{}, M\'arx \cite{marx:cut}
introduced the {\em{important separator}} technique,
which turned out to be very robust, and is now the key ingredient
in parameterized algorithms for various problems such as variants
of \fvsshort{}~\cite{directed-fvs,sfvs} or {\sc{Almost 2-SAT}} \cite{almost2sat-fpt}.
Moreover, the recent developments on \mwcshort{} show applicability of
linear programming in parameterized complexity, leading to the fastest currently known
algorithms not only for \mwcshort{}, but also \asatname{} and \octshort{}
\cite{nmc:2k,saket:lp}. Last but not least, the research on the \octshort{} problem resulted in the introduction of iterative compression, a simple yet powerful technique for designing parameterized algorithms \cite{reed:ic}.

\paragraph{Considered problem.} In this paper we study a robust generalization of the \fvsshort{} problem, namely \gfvs{}\footnote{In this paper, we follow the notation of Guillemot \cite{guillemot-journal}.}.
Let $\Sigma$ be a finite (not necessarily abelian) group, with unit element $1_\Sigma$.
We use the multiplicative convention for denoting the group operation.

\begin{definition}
For a finite group $\Sigma$, a directed graph $G=(V,A)$ and
a labeling function $\Lambda:A \rightarrow \Sigma$, we call $(G,\Lambda)$ 
a {\em $\Sigma$-labeled} graph iff for each arc $(u,v) \in A$ we have $(v,u) \in A$
and $\Lambda((u,v)) = \Lambda((v,u))^{-1}$.
\end{definition}

We somehow abuse the notation and by $(G\setminus X, \Lambda)$ denote
the $\Sigma$-labeled graph $(G, \Lambda)$ with vertices of $X$ removed,
even though formally $\Lambda$ has in its domain arcs that do not exist in $G\setminus X$.

For a path $P=(v_1,\ldots,v_{\ell})$ we denote $\Lambda(P)=\Lambda((v_1,v_2)) \cdot \ldots \cdot \Lambda((v_{\ell-1},v_{\ell}))$.
Similarly, for a cycle $C=(v_1,\ldots,v_{\ell},v_1)$ we denote $\Lambda(C)=\Lambda((v_1,v_2)) \cdot \ldots \cdot \Lambda((v_{\ell-1},v_{\ell})) \cdot \Lambda((v_{\ell},v_1))$.
We call a cycle $C$ a {\em non-null} cycle, iff $\Lambda(C)\not= 1_\Sigma$.
Observe that if the group $\Sigma$ is non-abelian, then it may happen that cyclic shifts of the same cycle yield different elements of the group; nevertheless, the notion of a non-null cycle is well-defined, as either all of them are equal to $1_\Sigma$ or none of them.

\begin{lemma}
Let $(x_1,\ldots,x_\ell,x_1)$ be a cycle in a $\Sigma$-labeled graph $(G,\Lambda)$.
If $\Lambda((x_1,\ldots,x_\ell,x_1)) \not= 1_\Sigma$, then $\Lambda((x_2,\ldots,x_{\ell},x_1,x_2))\not= 1_\Sigma$.
\end{lemma}

\begin{proof}
Let $g_1 = \Lambda((x_1,x_2))$ and $g_2=\Lambda((x_2,\ldots,x_{\ell},x_1))$. 
We have that $g_1 \cdot g_2 = 1_\Sigma$ iff $g_2 \cdot g_1 = 1_\Sigma$ and the lemma follows.
\end{proof}

In the \gfvs problem we want to hit all non-null cycles in a $\Sigma$-labeled graph using at most $k$ vertices.

\defproblemu{\gfvs (\gfvsshort{})}{A $\Sigma$-labeled graph $(G,\Lambda)$ and an integer $k$.}{$k$}{Does there exist a set $X \subseteq V(G)$ of at most $k$ vertices, such that there is no non-null cycle in $(G \setminus X, \Lambda)$?}

As observed in~\cite{guillemot-journal}, for a graph excluding a non-null cycle we can define a consistent labeling.

\begin{definition}
\label{def:consistent}
For a $\Sigma$-labeled graph $(G,\Lambda)$ we call $\lambda:V \rightarrow \Sigma$ a {\em consistent labeling} iff
for each arc $(u,v) = a \in A(G)$ we have $\lambda(v) = \lambda(u) \cdot \Lambda(a)$.
\end{definition}

\begin{lemma}[\cite{guillemot-journal}]
\label{lem:consistent}
A $\Sigma$-labeled graph $(G,\Lambda)$ has a consistent labeling iff it does not contain a non-null cycle.
\end{lemma}

Note that when analyzing the complexity of the \gfvsshort{} problem, it is important how the group $\Sigma$ is represented.
In~\cite{guillemot-journal} it is assumed that $\Sigma$ is given via its multiplication table
as a part of the input.
In this paper we assume a more general model, where operations in $\Sigma$ are computed
by an oracle in polynomial time.
More precisely, we assume that the oracle can multiply two elements, return an inverse of an element, provide the neutral element $1_\Sigma$, or check whether two elements are equal.

As noted in \cite{stefan:new}, \gfvsshort{} subsumes not only the classical \fvsshort{} problem,
but also \octshort{} (with $\Sigma = \mathbb{Z}_2$) and \mwcshort{} (with $\Sigma$ being an arbitrary group of size not smaller than the number of terminals).
We note that if $\Sigma$ is given in the oracle model, \gfvs subsumes also \esfvs, which is equivalent to \sfvs~\cite{sfvs}.

\defproblemu{\esfvs (\esfvsshort{})}{An undirected graph $G$, a set $S \subseteq E(G)$ and an integer $k$.}{$k$}{Does there exist a set $X \subseteq V(G)$ of at most $k$ vertices,
  such that in $G \setminus X$ there are no cycles with at least one edge from $S$?}

\begin{lemma}\label{lem:sfvs}
Given an \esfvsshort{} instance $(G,S,k)$, one can in polynomial time construct an equivalent \gfvsshort{} instance $(G',\Lambda,k)$ with group $\Sigma = \mathbb{Z}_2^{|S|}$.
\end{lemma}
\begin{proof}
To construct the new \gfvsshort instance, create the graph $G'$ by replacing each edge of $G$ with arcs in both direction,
 keep the parameter $k$, take $\Sigma = \mathbb{Z}_2^{|S|}$ and construct a $\Sigma$-labeling $\Lambda$ by setting
 any $|S|$ linearly independent values of $\Lambda((u,v))$ for $uv \in S$ and $\Lambda((u,v)) = 1_\Sigma$ for $uv \notin S$.
Clearly, this construction can be done in polynomial time and the operations on the group $\Sigma$ can be performed
by a polynomial-time oracle.
\end{proof}

We note that the \gfvs problem was also studied from the graph theoretical point of view, as, in addition to the aforementioned
reductions, it also subsumes the setting of Mader's $\mathcal{S}$-paths theorem \cite{maria:gfvs,kebab:gfvs}.
In particular, Kawarabayashi and Wollan proved the Erd\"{o}s-P\'{o}sa property for non-null cycles in highly connected graphs,
generalizing a list of previous results \cite{kebab:gfvs}.

The study of parameterized complexity of \gfvsshort{} was initiated by Guillemot \cite{guillemot-journal}, who presented
a fixed-parameter algorithm for \gfvsshort{} parameterized by $|\Sigma|+k$
running in time\footnote{The $\Ohstar()$ notation suppresses terms polynomial in the input size.}
$\Ohstar(2^{O(k\log |\Sigma|)})$.
When parameterized by $k$, Guillemot showed a fixed-parameter algorithm for
the easier edge-deletion variant of \gfvsshort{}, running in time $\Ohstar(2^{O(k \log k)})$.
Very recently, Kratsch and Wahlstr\"{o}m presented a randomized kernelization algorithm
that reduces the size of a \gfvsshort{} instance to $O(k^{2|\Sigma|})$ \cite{stefan:new}.

The main purpose of studying the \gfvsshort{} problem is to find the common points in the fixed-parameter algorithms for problems it generalizes. Precisely this approach has been presented by Guillemot in \cite{guillemot-journal}, where at the base of the algorithm lies a subroutine that solves a very general version of \mwc{}. When reducing various graph cut problems to \gfvsshort{}, usually the size of the group depends on the number of distinguished vertices in the instance, as in Lemma~\ref{lem:sfvs}. Hence, the usage of the general $\Ohstar(2^{O(k\log |\Sigma|)})$ algorithm of Guillemot unfortunately incorporates this parameter in the running time. It appears that by a more refined combinatorial analysis, usually one can get rid of this dependence; this is the case both in \sfvs~\cite{sfvs} and in \mwc~\cite{nmc:2k,saket:lp}. This suggests that the phenomenon can be, in fact, more general.

\paragraph{Our result and techniques.}
Our main result is a fixed-parameter algorithm for \gfvsshort{} parameterized by the
size of the cutset only.
\begin{theorem}
\label{thm:main}
\gfvs can be solved in $\Ohstar(2^{O(k\log k)})$ time and polynomial space.
\end{theorem}
Our algorithm uses a similar approach as described by Kratsch and Wahlstr\"{o}m
in \cite{stefan:new}: in each step of iterative compression, when we are given
a solution $Z$ of size $k+1$, we guess the values of a consistent labeling on the
vertices of $Z$, and reduce the problem to \mwc{}. However, by a straightforward
application of this approach we obtain $\Ohstar(2^{O(k \log |\Sigma|)})$ time complexity.
To reduce the dependency on $|\Sigma|$, we carefully analyze the structure of a solution,
provide a few reduction rules in a spirit of the ones used in the recent algorithm
for \sfvs{} \cite{sfvs} and, finally, for each vertex of $Z$
we reduce the number of choices for a value of a consistent labeling to polynomial in $k$.
Therefore, the number of reasonable consistent labelings of $Z$ is bounded
by $2^{O(k \log k)}$ and we can afford solving a \mwc{} instance for each such labeling. 

Note that the bound on the running time of our algorithm matches the currently best known algorithm for \sfvs{} \cite{sfvs}. Therefore, we obtain the same running time as in \cite{sfvs} by applying a much more general framework.

In the \gfvs problem definition in~\cite{guillemot-journal}
a set of forbidden vertices $F \subseteq V(G)$
is additionally given as a part of the input.
Observe that one can easily gadget such vertices by replacing each forbidden vertex by a clique
of size $k+1$ labeled with $1_\Sigma$; therefore, for the sake of simplicity we assume
that all the vertices are allowed.

\section{Preliminaries}

\paragraph{Notation.} We use standard graph notation.
For a graph~$G$, by~$V(G)$ and~$E(G)$ we denote its vertex and edge sets, respectively.
In case of a directed graph $G$, we denote the arc set of $G$ by $A(G)$.
For~$v \in V(G)$, its neighborhood~$N_G(v)$ is defined as~$N_G(v) = \{u: uv\in E(G)\}$, and~$N_G[v] = N_G(v) \cup \{v\}$ is the closed neighborhood of~$v$.
We extend this notation to subsets of vertices:~$N_G[X] = \bigcup_{v \in X} N_G[v]$ and~$N_G(X) = N_G[X] \setminus X$.
For a set~$X \subseteq V(G)$ by~$G[X]$ we denote the subgraph of~$G$ induced by~$X$.
For a set~$X$ of vertices or edges of~$G$, by~$G \setminus X$ we denote the graph
with the vertices or edges of~$X$ removed; in case of vertex removal, we remove
also all the incident edges.



\section{Algorithm}

In this section we prove Theorem \ref{thm:main}.
We proceed with a standard application of the iterative compression technique
in Section \ref{sec:alg:1}.
In each step of the iterative compression, we solve a \compgfvs{} problem, where
we are given a solution $Z$ of size a bit too large --- $k+1$ --- and we are to find a new
solution disjoint with it.
We first prepare the \compgfvs{} instance by {\em{untangling}} it in Section \ref{sec:alg:2},
  in the same manner as it is done in the kernelization algorithm of \cite{stefan:new}.
The main step of the algorithm is done in Section \ref{sec:alg:3},
 where we provide a set of reduction rules that enable us for each vertex
$v \in Z$ to limit the number of choices for a value of a consistent labeling on $v$
to polynomial in $k$. Finally, we iterate over all $\Ohstar(2^{O(k \log k)})$
remaining labelings of $Z$ and, for each labeling, reduce the instance to
\mwc{} (Section \ref{sec:alg:4}).

\subsection{Iterative compression}
\label{sec:alg:1}

The first step in the proof of Theorem~\ref{thm:main} is a standard technique in the design
of parameterized algorithms, that is, iterative compression, introduced by Reed et al.~\cite{reed:ic}.
Iterative compassion was also the first step of the parameterized algorithm for \sfvs~\cite{sfvs}.

We define a {\em compression problem}, where the input additionally contains a feasible 
solution $Z \subseteq V(G)$, and we are asked whether there exists a solution of size at most $k$
which is disjoint with $Z$. 

\defproblemugoal{\compgfvs (\cgfvsshort{})}{A $\Sigma$-labeled graph $(G,\Lambda)$, an integer $k$ and a set $Z \subseteq V(G)$, such that $(G\setminus Z, \Lambda)$ has no non-null cycle.}{$k+|Z|$}{Find a set $X \subseteq V(G) \setminus Z$ of at most $k$ vertices, such that there is no non-null cycle in $(G \setminus X, \Lambda)$ or return NO, if such a set does not exist.}

In Section~\ref{sec:alg:2} we prove the following lemma providing a parameterized algorithm for \compgfvs.

\begin{lemma}
\label{lem:compression}
\compgfvs can be solved in $\Ohstar(2^{O(|Z|(\log k + \log |Z|))} \cdot 2^k)$ time and polynomial space.
\end{lemma}

Armed with the aforementioned result, we can easily prove Theorem~\ref{thm:main}.

\begin{proof}[Proof of Theorem~\ref{thm:main}]
In the iterative compression approach we start with an empty solution for an empty graph, and in each of the $n$ steps we add a single vertex both to a feasible solution and to the graph; we use Lemma~\ref{lem:compression} to compress the feasible solution after guessing which vertices of the solution of size at most $k+1$ should not be removed.

Formally, for a given instance $(G=(V,A),\Lambda,k)$ let $V = \{v_1, \ldots, v_n\}$.
For $0 \le i \le n$ define $V_i = \{v_1,\ldots,v_i\}$ (in particular $V_0=\emptyset$)
and let $\Lambda_i$ be the function $\Lambda$ restricted to the set
of arcs $A_i=\{(u,v) \in A : u,v \in V_i\}$.
Initially we set $X_0=\emptyset$, which is a solution to the graph $(G[V_0],\Lambda_0)$.
For each $i=1,\ldots,n$ we set $Z_i=X_{i-1} \cup \{v_i\}$, which is a feasible solution
to $(G[V_i],\Lambda_i)$ of size at most $k+1$.
If $|Z_i| \le k$, then we set $X_i=Z_i$ and continue the inductive process.
Otherwise, if $|Z_i| = k+1$, we guess by trying all possibilities, a subset of vertices
$Z_i' \subseteq Z_i$ that is not removed in a solution of size $k$ to $(G[V_i],\Lambda_i)$
and use Lemma~\ref{lem:compression} for the instance $I_{Z_i'}=(G[V_i\setminus (Z_i \setminus Z_i')],\Lambda_i,k'=|Z_i'|-1,Z_i')$.
If for each set $Z_i'$ the algorithm from Lemma~\ref{lem:compression} returns NO, 
then there is no solution for $(G[V_i],\Lambda_i)$ and, consequently,
there is no solution for $(G,\Lambda)$.
However, if for some $Z_i'$ the algorithm from Lemma~\ref{lem:compression}
returns a set $X_i'$ of size smaller than $|Z_i'|$, then we set $X_i = (Z_i \setminus Z_i') \cup X_i'$.
Since $|X_i| = |Z_i \setminus Z_i'| + |X_i'| < |Z_i| = k+1$,
the set $X_i$ is a solution of size at most $k$ for the instance $(G_i,\Lambda_i)$.

Finally, we observe that since $(G_n,\Lambda_n)=(G,\Lambda)$, the set $X_n$
is a solution for the initial instance $(G=(V,A),\Lambda,k)$ of \gfvs. The claimed bound on running time follows from the observation that $|Z_i|\leq k+1$ for each of polynomially many steps.
\end{proof}

At this point a reader might wonder why we do not add an assumption $|Z| \le k+1$ to the \cgfvsshort problem
definition and parameterize the problem solely by $k$. 
The reason for this is that in Section~\ref{sec:alg:3} we will solve the \cgfvsshort problem
recursively, sometimes decreasing the value of $k$ without decreasing the size of $Z$,
and to always work with a feasible instance of the \cgfvsshort problem we avoid adding the $|Z| \le k+1$
assumption to the problem definition.

\subsection{Untangling}
\label{sec:alg:2}

In order to prove Lemma~\ref{lem:compression} we use the concept of {\em untangling},
previously used by Kratsch and Wahlstr\"om~\cite{stefan:new}.
We transform an instance of \cgfvsshort to ensure that each arc $(u,v)$
with both endpoints in $V(G) \setminus Z$ is labeled $1_\Sigma$ by $\Lambda$.

\begin{definition}
We call an instance $(G=(V,A),\Lambda,k,Z)$ of \cgfvsshort\ {\em untangled}, iff for each arc $(u,v) \in A$ such that $u,v \in V\setminus Z$
we have $\Lambda((u,v))=1_\Sigma$.

Moreover, by {\em untangling} a labeling $\Lambda$ around vertex $v$ with a group element $g$ we mean changing the labeling to $\Lambda':A \rightarrow \Sigma$, such that for $(u,v) = a \in A$, we have
$$\Lambda'(a)=
\left\{ \begin{array}{ll}
   g \cdot \Lambda(a) & {\textrm{if }} u = x;\\
   \Lambda(a) \cdot g^{-1}  & {\textrm{if }} v = x; \\
   \Lambda(a) & {\textrm{otherwise}}.
\end{array} \right. $$
\end{definition}

\begin{lemma}
\label{lem:relabel}
Let $(G=(V,A),\Lambda)$ be a $\Sigma$-labeled graph, $x \in V$ be a vertex of $G$ and let $g \in \Sigma$ be a group element.
For any subset of vertices $X \subseteq V$ the graph $(G\setminus X,\Lambda)$ contains a non-null cycle iff $(G\setminus X,\Lambda')$ contains a non-null cycle, where $\Lambda'$ is
the labeling $\Lambda$ untangled around the vertex $x$ with a group element~$g$.
\end{lemma}

\begin{proof}
The lemma follows from the fact that for any cycle $C$ in $G$ we have $\Lambda(C)=\Lambda'(C)$.
\end{proof}

In Section~\ref{sec:alg:3} we prove the following lemma.

\begin{lemma}
\label{lem:untangled}
\compgfvs for untangled instances can be solved in $\Ohstar(2^{O(|Z| (\log k + \log |Z|))} \cdot 2^k)$ time and polynomial space.
\end{lemma}

Having Lemmata~\ref{lem:relabel} and \ref{lem:untangled} we can prove Lemma~\ref{lem:compression}.

\begin{proof}[Proof of Lemma~\ref{lem:compression}]
Let $(G,\Lambda,k,Z)$ be an instance of \cgfvsshort.
Since $(G\setminus Z)$ has no non-null cycle, by Lemma~\ref{lem:consistent}
there is a consistent labeling $\lambda$ of $(G\setminus Z, \Lambda)$.

Let $\Lambda'$ be a result of untangling $\Lambda$ around each
vertex $v \in V(G)\setminus Z$ with $\lambda(v)$. 
Note that, by associativity of $\Sigma$, the order in which we untangle subsequent vertices does not matter.
After all the untangling operations, for an arc $a=(u,v) \in A(G)$,
such that $u,v \in V(G) \setminus Z$, we have $\Lambda'(a) = (\lambda(u) \cdot \Lambda(a)) \cdot \lambda(v)^{-1} = \lambda(v) \cdot \lambda(v)^{-1} = 1_\Sigma$.
Therefore, by Lemma~\ref{lem:relabel} the instance $(G,\Lambda',k,Z)$
is an untangled instance of \cgfvsshort, which is a YES-instance
iff $(G,\Lambda,k,Z)$ is a YES-instance.
Consequently, we can use Lemma~\ref{lem:untangled} and the claim follows.
\end{proof}

\subsection{Fixing a labeling on $Z$}
\label{sec:alg:3}

In this section we prove Lemma~\ref{lem:untangled} using the following lemma, which we prove in Section~\ref{sec:alg:4}.

\begin{lemma}
\label{lem:fixed}
Let $(G,\Lambda,k,Z)$ be an untangled instances of \cgfvsshort.
There is an algorithm which for a given
function $\phi:Z\rightarrow \Sigma$,
finds a set $X \subseteq V(G)\setminus Z$ of size at most $k$,
such that there exists a consistent labeling 
$\lambda:V(G)\setminus X \rightarrow \Sigma$ of $(G\setminus X,\Lambda)$,
where $\lambda|_{Z} = \phi$,
or checks that such a set $X$ does not exist; the algorithm works in $\Ohstar(2^k)$ time and uses polynomial space.
\end{lemma}

We could try all $(|\Sigma|+1)^{|Z|}$
possible assignments $\phi$ and use the algorithm from Lemma~\ref{lem:fixed}.
Unfortunately, since $|\Sigma|$ is not our parameter we cannot 
iterate over all such assignments. Therefore, the goal of
this section is to show that after some preprocessing,
it is enough to consider only $2^{O(|Z|(\log k + \log |Z|))}$ assignments $\phi$; together with Lemma~\ref{lem:fixed} this suffices to prove Lemma~\ref{lem:untangled}.

\begin{definition}
Let $(G,\Lambda,k,Z)$ be an untangled instance of \cgfvsshort,
let $z$ be a vertex in $Z$ and by $\Sigma_z$ denote 
the set $\Lambda(\{(z,v) \in A(G) : v \in V(G) \setminus Z\})$.

By a flow graph $F(G,\Lambda,Z,z)$, we denote the undirected
graph $(V',E')$, where $V' = (V(G) \setminus Z) \cup \Sigma_z$
and $E' = \{uv : (u,v) \in A(G[V(G) \setminus Z])\} \cup \{gv : (z,v) \in A(G), v \in V(G) \setminus Z, \Lambda((z,v)) = g\}$.
\end{definition}

Less formally, in the flow graph we take the underlying undirected graph of $G[V(G) \setminus Z]$
and add a vertex for each group element $g \in \Sigma_z$, 
that is a group element for which there exists an arc from $z$ to $V(G) \setminus Z$ 
labeled with $g$ by $\Lambda$.
A vertex $g \in \Sigma_z$ is adjacent to all the vertices of $V(G) \setminus Z$
for which there exists an arc going from $z$, labeled with $g$ by $\Lambda$.

\begin{lemma}
\label{red:1}
Let $(G,\Lambda,k,Z)$ be an untangled instance of \cgfvsshort.
Let $H$ be the flow graph $F(G,\Lambda,Z,z)$ for some $z \in Z$.
If for some vertex $v \in V(G)\setminus Z$, 
in $H$ there are at least $k+2$ paths from $v$ to $\Sigma_z$ that are vertex disjoint apart from $v$,
then $v$ belongs to every solution of \cgfvsshort.
\end{lemma}

\begin{proof}
Let us assume, that $v$ is not a part of a solution $X \subseteq V(G)\setminus Z$,
where $|X| \le k$.
Then there at least $2$ out of the $k+2$ paths from $v$ to $\Sigma_z$ remain in $H \setminus X$.
These two paths are vertex disjoint apart from $v$, so they correspond to a non-null cycle in $G\setminus X$, a contradiction.
\end{proof}

\begin{definition}
For an untangled instance $(G,\Lambda,k,Z)$ of \cgfvsshort by an {\em external path}
we denote any path $P$ beginning and ending in $Z$, but with all internal vertices belonging to $V(G)\setminus Z$.
Moreover, for two distinct vertices $z_1,z_2 \in Z$ by $\Sigma(z_1,z_2)$ we denote the set of all elements $g \in \Sigma$,
for which there exists an external path $P$ from $z_1$ to $z_2$ with $\Lambda(P)=g$.
\end{definition}

\begin{lemma}
\label{lem:many_paths}
Let $(G,\Lambda,k,Z)$ be an untangled instance of \cgfvsshort.
If for each $z \in Z$ and $v \in V(G) \setminus Z$
there are at most $k+1$ vertex disjoint paths from $v$
to $\Sigma_z$ in $F(G,\Lambda,Z,z)$
and for some $z_1,z_2 \in Z$, $z_1 \neq z_2$, we have $|\Sigma(z_1,z_2)| \ge k^3(k+1)^2+2$, 
then there is no solution for $(G,\Lambda,k,Z)$.
\end{lemma}

\begin{proof}
Let us assume that $X \subseteq V(G)\setminus Z$ is a solution for $(G,\Lambda,k,Z)$.
Let $\mathcal{P}$ be a set of external paths from $z_1$ to $z_2$,
containing exactly one path $P$ for each $g \in \Sigma(z_1,z_2)$ with $\Lambda(P)=g$.
Note that the only arcs with non-null labels in $P$ are possibly the first and the last arc.

By the pigeon-hole principle, there exists a vertex $v \in X$,
which belongs to at least $k^2(k+1)^2+1$ paths in $\mathcal{P}$,
since otherwise there would be at least two paths in $\mathcal{P}$
disjoint with $X$, creating a non-null cycle disjoint with $X$. This cycle is not necessarily simple; however, if it is non-null, then it contains a simple non-null subcycle that is also disjoint with $X$.

Consider a connected component $C$ of $G[V(G) \setminus Z]$ to which $v$ belongs.
Observe that there exists a vertex $z\in \{z_1,z_2\}$
that has at least $k(k+1)+1$ incident arcs
going to $C$ with pairwise different labels in $\Lambda$,
since otherwise $v$ would belong to at most $k^2(k+1)^2$ paths in $\mathcal{P}$.

Let $H$ be the flow graph $F(G,\Lambda,Z,z)$
and let $T \subseteq \Sigma_z$ be the set of labels of arcs going from $z$ to $C$; recall that $|T| > k(k+1)$.
Since there is no non-null cycle in $(G\setminus X,\Lambda)$, 
we infer that in $H_0 = H[C \cup T] \setminus (X\cap C)$, 
no two vertices of $T$ belong to the same connected component.
Moreover, as $C$ is connected in $G$, for each $t \in T$ there exists a path $P_t$ with endpoints $v$ and $t$ in $H[C \cup T]$.
Let $w_t$ be the closest to $t$ vertex from $X$ on the path $P_t$. As $|X| \leq k$ and $|T| > k(k+1)$,
there exists $w \in X$ such that $w = w_t$ for at least $k+2$ elements $t \in T$.
By the definition of the vertices $w_t$ and the fact that there are no two vertices of $T$ in the same connected component of $H_0$,
the subpaths of $P_t$ from $t$ to $w_t$ for all $t$ with $w=w_t$ are vertex disjoint apart from $w$.
As there are at least $k+2$ of them, we have a contradiction.
\end{proof}

We are now ready to prove Lemma~\ref{lem:untangled} given Lemma \ref{lem:fixed}.

\begin{proof}[Proof of Lemma~\ref{lem:untangled}]
If there exists a vertex $v$, satisfying the properties of Lemma~\ref{red:1},
we can assume that it has to be a part of the solution; therefore,
we can remove the vertex from the graph and solve the problem
for decremented parameter value.
Hence, we assume that for each $z \in Z$ and $v \in V(G) \setminus Z$,
there are at most $k+1$ vertex disjoint paths from $v$ to $\Sigma_z$ in $F(G,\Lambda,Z,z)$.
We note that one can compute the number of such vertex disjoint paths in polynomial time, using 
a maximum flow algorithm.

By Lemma~\ref{lem:many_paths}, if there is a pair of vertices $z_1,z_2 \in Z$ with $|\Sigma(z_1,z_2)| \ge k^3(k+1)^2+2$, 
we know that there is no solution.
Observe, that one can easily verify the cardinality of $\Sigma(z_1,z_2)$, since the only
non-null label arcs on paths contributing to $\Sigma(z_1,z_2)$ are the first and the last one,
and we can iterate over all such arcs and check whether their endpoints are in the same connected component in $G[V(G) \setminus Z]$.
Clearly, this can be done in polynomial time.

Knowing that the sets $\Sigma(z_1,z_2)$ have sizes bounded by a function of $k$,
we can enumerate all the reasonable labelings of $Z$.
For the sake of analysis let $G'=(Z,E')$ be an auxiliary undirected graph,
where two vertices of $Z$ are adjacent, when they are connected by an external path
in $G\setminus X$, for some fixed solution $X \subseteq V(G) \setminus Z$.
Let $F$ be any spanning forest of $G'$.
Since $F$ has at most $|Z|-1$ edges, we can guess $F$,
by trying at most $|Z| \cdot |Z|^{2(|Z|-1)}$ possibilities.
Let us assume, that we have guessed $F$ correctly.
Observe that for any two vertices $z_1, z_2 \in Z$,
belonging to two different connected components of $F$,
there is no path between $z_1$ and $z_2$ in $G \setminus X$.
Therefore, there exists a consistent labeling of $G \setminus X$,
which labels an arbitrary fixed vertex from each connected component of $F$ with $1_\Sigma$.
For all other vertices of $F$ we use the fact that if we have already fixed a value $\phi(z_1)$,
then for each external path corresponding to an edge $z_1z_2$ of $F$, there are at most $k^3(k+1)^2+1$ possible 
values of $\phi(z_2)$, since $\phi^{-1}(z_1) \cdot \phi(z_2) \in \Sigma(z_1,z_2)$.
Hence, we can exhaustively try $2^{O(|Z|(\log k + \log |Z|))}$ labelings $\phi$ of $Z$,
and use Lemma~\ref{lem:fixed} for each of them.
\end{proof}

\subsection{Reduction to Multiway Cut}
\label{sec:alg:4}

In this section, we prove Lemma~\ref{lem:fixed},
by a reduction to \mwc.
A similar reduction was also used recently by Kratsch and Wahlstr\"om
in the kernelization algorithm for \gfvs parameterized by $k$ with constant $|\Sigma|$~\cite{stefan:new}.
Currently the fastest FPT algorithm for \mwc is due to Cygan et al.~\cite{nmc:2k},
and it solves the problem in $\Ohstar(2^k)$ time and polynomial space.

\defproblemugoal{\mwc}{An undirected graph $G = (V,E)$, a set of terminals $T \subseteq V$, and a positive integer $k$.}{$k$}
{Find a set $X \subseteq V\setminus T$,
  such that $|X| \leq k$ and no pair of terminals from the set $T$ is contained in one connected component of the graph $G[V\setminus X]$,
  or return NO if such a set $X$ does not exist.}

\begin{proof}[Proof of Lemma~\ref{lem:fixed}]
Firstly, we check whether the given function $\phi$ satisfies $\phi(z_2) = \phi(z_1) \cdot \Lambda((z_1,z_2))$,
for each arc $(z_1,z_2) \in G[Z]$, since otherwise there is no set $X$ we are looking for.

Given a $\Sigma$-labeled graph $(G,\Lambda)$,
a set $Z$, an integer $k$, and a function $\phi:Z \rightarrow \Sigma$,
we create an undirected graph $G'=(V,E)$.
As the vertex set, we set $V=(V(G) \setminus Z) \cup T$ and $T=\{g : (u,v) \in A(G),\ u\in Z,\ v \in V(G) \setminus Z,\ \phi(u) \cdot \Lambda((u,v)) = g\}$.
Note that in the set $T$ there exactly these elements of $\Sigma$,
which are potential values of a consistent labeling
of $(G,\Lambda)$ that matches $\phi$ on $Z$.
As the edge set, we set $E=\{uv : (u,v) \in A(G[V(G) \setminus Z])\} \cup \{gv: (u,v) \in A(G),\ u\in Z,\ v\in V(G) \setminus Z,\ \phi(u) \cdot \Lambda((u,v))=g\}$.
We show that $(G',T,k)$ is a YES-instance of \mwc iff 
there exists a set $X\subseteq V(G)\setminus Z$, such that there exists a consistent labeling $\lambda$ of $(G\setminus X, \Lambda)$ with $\lambda|_Z=\phi$.

Let $X$ be solution for $(G',T,k)$. We define a consistent labeling $\lambda$ of $(G\setminus X,\Lambda)$.
For $v \in Z$ we set $\lambda(v) = \phi(v)$. 
For $v \in (V(G) \setminus Z) \setminus X$, if $v$ is reachable from a terminal $g \in T$ in $G'\setminus X$, we set $\lambda(v) = g$.
If $v \in (V(G) \setminus Z) \setminus X$ is not reachable from any terminal in $G'$, we set $\lambda(v) = 1_\Sigma$.
Since each arc in $A(G[V(G) \setminus Z])$ is labeled $1_\Sigma$ by $\Lambda$, 
and each vertex in $V(G) \setminus Z$ is reachable from at most one terminal in $G'\setminus X$,
$\lambda$ is a consistent labeling of $(G\setminus X,\Lambda)$.

Let $X\subseteq V(G)\setminus Z$ be a set of vertices of $G$, $|X| \le k$,
such that there is a consistent labeling $\lambda$ of $(G\setminus X,\Lambda)$,
where $\lambda|_Z=\phi$.
By the definition of edges between $T$ and $V(G) \setminus Z$ in $G'$,
each vertex of $V(G) \setminus Z$ is reachable from at most one terminal in $G'$,
since otherwise $\lambda$ would not be a consistent labeling of $(G\setminus X,\lambda)$.
Therefore, $X$ is a solution for $(G',T,k)$.

We can now apply the algorithm for \mwc of~\cite{nmc:2k} to the instance $(G',T,k)$ in order to conclude the proof.
\end{proof}

\section{Conclusions and open problems}

We have shown a relatively simple fixed-parameter algorithm
for \gfvs{} running in time $\Ohstar(2^{O(k \log k)})$. Our algorithm
works even in a robust oracle model, that allows us to generalize
the recent algorithm for \sfvs{} \cite{sfvs} within the same complexity bound.

We would like to note that if we represent group elements
by strings consisting $g$ and $g^{-1}$ for $g \in \Lambda(A(G))$ (formally, we perform the computations in the free group over generators corresponding to the arcs of the graph), then after slight modifications
of our algorithm we can solve the \gfvs problem
even for infinite groups for which the word problem, i.e., the problem of checking whether results of two sequences of multiplications are equal, is polynomial-time solvable. The lengths of representations of group elements created during the computation can be bounded linearly in the size of the input graph. Therefore, if a group admits a polynomial-time algorithm solving the word problem, then we can use this algorithm as the oracle.

Both our algorithm and the algorithm for \sfvs{} of \cite{sfvs} seems
hard to speed up to time complexity $\Ohstar(2^{O(k)})$.
Can these problems be solved in $\Ohstar(2^{O(k)})$ time, or can we prove that such a result
would violate Exponential Time Hypothesis?

\vskip 0.5cm

\noindent{\bf{Acknowledgements.}} We thank Stefan Kratsch and Magnus Wahlstr\"om for inspiring discussions on graph separation problems and for drawing
our attention to the \gfvs{} problem.

\bibliographystyle{plain}
\bibliography{group-fvs}

\newpage

\end{document}